\newif\iffullversion
\fullversionfalse

\documentclass[submission,copyright,creativecommons]{eptcs}

\providecommand{\event}{GandALF 2026}

\usepackage{mathtools}
\usepackage{hyperref}
\usepackage{etoolbox}
\usepackage{amssymb}
\usepackage{amsthm}
\usepackage{thmtools} 
\usepackage{thm-restate}
\newtheorem{definition}{Definition}
\newtheorem{lemma}{Lemma}

\newtheorem{remark}{Remark}
\newtheorem{proposition}{Proposition}
\newtheorem{example}{Example}
\newtheorem{corollary}{Corollary}

\makeatletter
\patchcmd{\@bibitem}{\ignorespaces}{\label{bib-#1}\ignorespaces}{}{}
\makeatother

\DeclareMathOperator{\ltlX}{\mathsf{X}}
\DeclareMathOperator{\ltlU}{\mathsf{U}}
\DeclareMathOperator{\ltlW}{\mathsf{W}}
\DeclareMathOperator{\ltlR}{\mathsf{R}}
\DeclareMathOperator{\ltlM}{\mathsf{M}}
\DeclareMathOperator{\ltlS}{\mathsf{S}}
\DeclareMathOperator{\ltlWS}{\mathsf{\tilde{S}}}
\DeclareMathOperator{\ltlWY}{\mathsf{\tilde{Y}}}
\DeclareMathOperator{\ltlNS}{\mathsf{\tilde{B}}}
\DeclareMathOperator{\ltlG}{\mathsf{G}}
\DeclareMathOperator{\ltlF}{\mathsf{F}}
\DeclareMathOperator{\ltlY}{\mathsf{Y}}
\DeclareMathOperator{\ltlO}{\mathsf{O}}
\DeclareMathOperator{\ltlH}{\mathsf{H}}
\DeclareMathOperator{\ltlNWS}{\mathsf{B}}
\DeclareMathOperator{\op}{\mathsf{op}}

\newcommand{\rw}[1]{\langle #1 \rangle}
\newcommand{\wc}[0]{\mathsf{wc}}
\newcommand{\cwc}[0]{\delta_\mathcal{H}}
\newcommand{\sff}[0]{\mathsf{sf}}
\newcommand{\psf}[0]{\mathsf{psf}}
\newcommand{\af}[1][\sigma]{\mathsf{af}_{#1}}
\newcommand{\afl}[1][\sigma]{\mathsf{afl}_{#1}}
\newcommand{\afc}[1][\sigma]{\mathsf{afl}_{#1}}

\newcommand{\G}[3]{\mathcal{G}_{#2, #3}^{#1}}
\newcommand{\F}[3]{\mathcal{F}_{#2, #3}^{#1}}
\newcommand{\GF}[3]{\mathcal{GF}_{#2, #3}^{#1}}
\newcommand{\FG}[3]{\mathcal{FG}_{#2, #3}^{#1}}
\newcommand{\Var}{\mathsf{Var}}
\newcommand{\placeholder}{\makebox[1ex]{\textbf{$\cdot$}}}
\newcommand{\entcon}[2]{\mathcal{C}_{#1, #2}}
\newcommand{\entseq}[2]{\vec{\mathcal{C}}_{#1, #2}}
\newcommand{\compseq}[2]{\circ \, \entseq{#1}{#2}}
\newcommand{\entconw}[3]{\mathcal{C}^{#3}_{#1, #2}}
\newcommand{\entseqw}[3]{\vec{\mathcal{C}}^{#3}_{#1, #2}}
\newcommand{\compseqw}[3]{\circ \, \entseqw{#1}{#2}{#3}}

\def\titlerunning{New Translations From LTL with Past to LDBA and DPA}
\def\authorrunning{D. Lidell, N. Piterman}

\begin{document}
    
\title{New Translations From LTL with Past to LDBA and DPA%
\thanks{Supported in part by the Swedish research council (VR) project no. 2025-05419 and the Wallenberg AI, Autonomous Systems and Software Program (WASP) funded by the Knut and Alice Wallenberg Foundation.}
}
\author{David Lidell \qquad \qquad Nir Piterman\institute{University of Gothenburg and Chalmers University of Technology}}
%\author{David Lidell\orcidID{0000-0003-2760-339X}
%\and
%Nir Piterman\orcidID{0000-0002-8242-5357}
%}
%\institute{University of Somewhere}
%\institute{University of Gothenburg and Chalmers University of Technology\\
%\email{\{lidelld,piterman\}@chalmers.se}}

\maketitle

\begin{abstract}
We present asymptotically optimal translations from linear temporal logic with past initially to limit-deterministic Büchi automata and then onwards to deterministic parity automata. In contrast to previously known translations involving past, the ones presented here are incremental: states and transitions can be constructed as needed on-the-fly, and they do not involve backtracking in order to correct for temporal inconsistencies. The past is naturally encoded in the formula by alternating between weak and strong versions of past operators. To handle the difficulties that arise from the alternation of least and greatest fix-point operators (such as between the ``until'' operator and its weak counterpart), we rewrite the given formula based on the infinitary behavior of a given word with respect to subformulae rooted with such operators.

%\keywords{Linear temporal logic, Limit-deterministic automata, Non-deterministic automata}
\end{abstract}

\section{Introduction}
Linear temporal logic (\textsf{LTL}), introduced in the 1970s by Pnueli~\cite{pnueli_ltl}, is a logic of particular interest to the computer science community. This is both because it is suitable for expressing how the state of a system evolves over time and because of the perceived balance between its expressivity and complexity. Translations from \textsf{LTL} to $\omega$-automata have been extensively studied, not only because they provide a deeper understanding of the semantics of \textsf{LTL}, but also for practical reasons. Common use cases of $\omega$-automata include model checking and reactive synthesis, with the desired type of automaton depending on the situation.

In the context of model checking, non-deterministic $\omega$-automata combined with program models in the form of transition systems allow for decidable and efficient algorithms for verifying properties with respect to the model. For models of systems that exhibit probabilistic behavior, more determinism is required from the automaton, where limit-deterministic automata become useful. 
The latter make only finitely many non-deterministic choices and along every accepting run become deterministic.
In both cases an ``on-the-fly'' construction is desirable, where states of the automaton are incrementally built as needed for verification, potentially eliminating the need to construct the full state space.

One of the most successful approaches for both the non-deterministic and deterministic construction mentioned above has been \emph{semantic} \cite{Duret-LutzRCRAS22,strix}.
In this approach, each state of the respective automaton is a formula giving semantic information about the requirements for the future in a given state.
Recently, a semi-symbolic construction of automata for \textsf{LTLf} \cite{DuretLutzZPGV25} used this approach to win the \textsf{LTLf} track in the last synthesis competition.
It is also used in the synthesis from \textsf{LTL} to combine machine learning into the strategy search \cite{KretinskyMPZ25}, which won the \textsf{LTL} track of the same competition. 

Most existing translations are only capable of handling future temporal operators. This is perhaps based on the fact that linear temporal logic with past operators (\textsf{pLTL}) is no more expressive than \textsf{LTL}, and that Gabbay~\cite{Gabbay87} provided an algorithm for translating a given \textsf{pLTL} formula into an initially equivalent formula without past operators. Moreover, many properties are more naturally and elegantly expressed using past operators, and Markey~\cite{markey_succinct} showed that \textsf{pLTL} is exponentially more succinct than \textsf{LTL}.
Interestingly, the usage of \textsf{LTLf} in planning also led to interest in the incorporation of pure-past \textsf{LTL} \cite{AminofGRV25,BonassiGFFGS25} and considering pure-past \textsf{LTL} also for MDP verification \cite{Giacomo0SWY25}.

In Azzopardi et al.~\cite{azzopardi_et_al:LIPIcs.MFCS.2024.13} an asymptotically optimal semantic translation from \textsf{pLTL} to deterministic Rabin automata is presented. The translation is an extension of a similar result for \textsf{LTL} by Esparza et al.~\cite{esparza_unified_translation}, who also present a variation thereof for the construction of limit-deterministic Büchi automata. In this paper we show how to adapt this variation to handle past operators, combining the frameworks developed in both papers. As with its \textsf{LTL} counterpart, the translation is incremental in the sense that states can be created as needed, and its number of states is asymptotically optimal relative to the size of the input formula. We also show that the resulting automaton satisfies conditions that allow for a subsequent translation to a deterministic parity automaton without further blowups in size, as described by Esparza et al.~\cite{esparza2022dpa}.

%A unique advantage of our approach, where the history is incorporated into the formula, is that past subformulae that become irrelevant for the future evaluation simply disappear from the states, leading to more efficient constructions. 
Prior to our work, obtaining a limit-deterministic Büchi automaton or deterministic parity automaton for a given \textsf{pLTL} formula involved first constructing an intermediate non-deterministic Büchi automaton and from that deriving an equivalent automaton with the desired acceptance condition.

%\subsubsection*{Paper Overview}
The paper is structured as follows. Section~\ref{section_preliminaries} introduces the basic concepts behind $\omega$-automata and the relevant types of automata, as well as the syntax and semantics of \textsf{pLTL}. Section~\ref{section_previous_work} includes definitions and properties from previous work that are required to understand the translations. We omit repeating proofs and refer to the original papers. The translation is based on the after-function, which reads an input symbol and rewrites the input formula based on the information contained therein; the result is a formula to be verified by further input.  

As we allow for the arbitrary alternation of past and future operators, the current input may contain information relevant for the proper evaluation of past subformulae in future steps. In Section~\ref{section_encoding_past} we describe how to encode this information in the formula, while Section~\ref{section_after} presents a simplified version of the after-function for \textsf{pLTL}. We present the translation from \textsf{pLTL} to limit-deterministic Büchi automata in Section~\ref{section_pltl_to_ldba}, and show that the resulting automata are amenable to further asymptotically optimal translations to deterministic parity automata in Section~\ref{section:dpa}. We finish with concluding remarks in Section~\ref{section_discussion}.

\section{Preliminaries}
\label{section_preliminaries}
For readability, we enclose tuples either in parentheses or in angled brackets depending on context. For example, an ordered pair is either $(a, b)$ or $\langle a, b \rangle$. Given a set $S$, we denote by $\mathcal{P}(S)$ its power set. Given a natural number $n$, we write $[n]$ to represent the set $\{ 1, 2, \dots, n \}$.
\subsection{Infinite Words and \texorpdfstring{$\omega$}{Omega}-automata}
\label{section_preliminaries_automata}
Given a non-empty finite set of letters $\Sigma$, a \textit{word} over $\Sigma$ is an infinite sequence $\sigma_0, \sigma_1, \dots$, where $\sigma_t \in \Sigma$ for all $t \geq 0$. We call $\Sigma$ an \textit{alphabet} and write $\Sigma^\omega$ to denote the set of all words over $\Sigma$. Given a word $w = \sigma_0, \sigma_1, \dots$ we write $w_t = \sigma_t, \sigma_{t+1}, \dots$ to denote the suffix of $w$ starting at $t$. Given two indices $t$ and $s$, where $s \geq t$, the infix $\sigma_t, \sigma_{t+1}, \ldots, \sigma_{s-1}$ of $w$ length $s - t$ is denoted $w_{ts}$. If $t = s$ then $w_{ts} = \epsilon$, where $\epsilon$ is the empty word. We use parentheses to avoid ambiguity when necessary. For example, to denote $\sigma_{t+1}, \sigma_{t+2}, \dots, \sigma_s$ we write $w_{(t+1)(s+1)}$. The same notation is used throughout for suffixes and infixes of general infinite and finite sequences.

An $\omega$-\textit{automaton} $\mathcal{A}$ over an alphabet $\Sigma$ is a quadruple $(Q, Q_0, \delta, \alpha)$ where $Q$ is a finite set of states, $Q_0 \subseteq Q$ a set of initial states, and $\delta : Q \times \Sigma \to \mathcal{P}(Q)$ is a transition function. Finally, $\alpha$ is the automaton's acceptance condition. An $\omega$-automaton is \textit{deterministic} if $\vert Q_0 \vert = 1$ and $\vert \delta(q, \sigma) \vert = 1$ for all $q \in Q$ and $\sigma \in \Sigma$. For deterministic $\omega$-automata we treat $Q_0$ and $\delta(q, \sigma)$ as single states, omitting set notation. Moreover, we overload the notation to finite words $w = \sigma_0, \dots, \sigma_n$, writing $\delta(q, w)$ to denote the composition $\delta(\delta(\dots(\delta(q, \sigma_0)), \sigma_{n-1}), \sigma_n)$, with the convention that $\delta(q, \epsilon) = q$. We restrict our attention to $\omega$-automata and use ``automaton'' and ``$\omega$-automaton'' interchangeably.

A \textit{run} of an $\omega$-automaton $\mathcal{A}$ on a word $w = \sigma_0, \sigma_1, \dots$ is an infinite sequence $\eta = q_0, q_1, \dots$ such that $q_0 \in Q_0$ and $q_{t + 1} \in \delta(q_t, \sigma_t)$ for all $t \geq 0$. \textit{Büchi automata} are $\omega$-automata whose acceptance condition $\alpha$ is a set of accepting states. A run $\eta$ of a Büchi automaton $\mathcal{A}$ on a word $w$ is an \textit{accepting run} if $\mathit{Inf}(\eta) \cap \alpha \neq \varnothing$, where $\mathit{Inf}(\eta) \subseteq Q$ denotes the states of $Q$ that appear infinitely often in $\eta$. A word $w$ is \textit{accepted} by $\mathcal{A}$ if there exists an accepting run of $\mathcal{A}$ on $w$.
A Büchi automaton is \textit{limit-deterministic} if there exists a subset $Q_d \subseteq Q$ such that $\alpha \subseteq Q_d$ and, for all $q \in Q_d$ and all $\sigma \in \Sigma$, it holds that $\delta(q, \sigma) \subseteq Q_d$ and $\vert \delta(q, \sigma) \vert = 1$. We call the set $Q_d$ the \textit{accepting part} of the automaton, and its complement $Q_{nd}=Q \setminus Q_d$ the \textit{non-deterministic part}.
%A Büchi automaton is \textit{limit-deterministic} if, for every accepting state $q \in \alpha$, there exists a subset $Q_q \subseteq Q$ such that $q \in Q_q$ and for all $q' \in Q_q$ and $\sigma \in \Sigma$ it holds that $\delta(q', \sigma) \subseteq Q_q$ and $\vert \delta(q', \sigma) \vert = 1$. We call the set $Q_d$ the \textit{accepting part} of the automaton, and its complement $Q \setminus Q_d$ the \textit{non-deterministic part}.
\textit{Parity automata} are $\omega$-automata whose acceptance condition $\alpha$ is a priority function $\alpha: Q \rightarrow \mathbb{N}$ associating a ranking to each state.
A run $\eta$ of a parity automaton is accepting if $\min\{\alpha(q) \mid q\in \mathit{Inf}(\eta)\}$ is even. 

A \textit{multiple-entry $\omega$-automaton} is an automaton without distinguished initial states. These are used as intermediate automata as part of our construction in Section~\ref{section_pltl_to_ldba}. Accepting runs on multiple-entry automata are defined with respect to a given starting state.

A \textit{bed automaton} is an automaton without an acceptance condition: it is a triple $\mathcal{A} = (S, S_0, \delta_\mathcal{A})$ with set of states $S$, set of initial states $S_0 \subseteq S$, and transition function $\delta_\mathcal{A} : S \times \Sigma \to \mathcal{P}(S)$. A \textit{runner automaton}, defined with respect to a bed automaton $\mathcal{A}$, is an automaton $\mathcal{B} = (Q, Q_0, \delta, \alpha)$ with a specialized transition function $\delta : Q \times S \times \Sigma \to \mathcal{P}(Q)$ and a B\"uchi acceptance condition $\alpha$. Finally, the composition $\mathcal{A} \ltimes \mathcal{B}$ of a bed automaton $\mathcal{A}$ and runner automaton $\mathcal{B}$ defined with respect to $\mathcal{A}$ is the B\"uchi automaton $(Q \times S, Q_0 \times S_0, \overline{\delta}, \overline{\alpha})$, where $\overline{\delta}(q, s, \sigma) = \{ \langle q', s' \rangle \mid q' \in \delta(q, s', \sigma), s' \in \delta_\mathcal{A}(s, \sigma) \}$ and $\overline{\alpha} = \alpha \times S$.
This is a specialized form of cascade composition of automata.\footnote{See explanation regarding the parts of a \href{https://en.wikipedia.org/wiki/Millstone}{millstone} for the intuition regarding these names.}

The \textit{language} of an automaton $\mathcal{A}$ over an alphabet $\Sigma$ is the set $\mathcal{L}(\mathcal{A}) \subseteq \Sigma^\omega$ of words over $\Sigma$ that are accepted by $\mathcal{A}$. Given a state $q$ of $\mathcal{A}$, we similarly write $\mathcal{L}(q)$ to denote the set of words accepted by $\mathcal{A}$ when $q$ is taken to be the single initial state.
We use the acronyms NBA and LDBA to denote non-deterministic Büchi automata and limit-deterministic Büchi automata, respectively. We write DPA to denote deterministic parity automata. When the type of automaton is clear from the context, we simply write ``automaton''.

\subsection{Linear Temporal Logic with Past}
Given a finite set of propositional variables $\mathit{AP}$, a well-formed formula $\varphi$ of \textsf{pLTL} is generated by the grammar
\begin{equation*}
\begin{aligned}
    \varphi &\Coloneqq \top \mid \bot \mid p \mid \neg \varphi \mid \varphi \land \varphi \mid \varphi \lor \varphi \mid \ltlX \varphi \mid \varphi \ltlU \varphi \mid \ltlY \varphi \mid \varphi \ltlS \varphi,
\end{aligned}
\end{equation*}
where $p \in \mathit{AP}$. When a word $w = \sigma_0, \sigma_1, \dots\in (\mathcal{P}(\mathit{AP}))^\omega$ \textit{models} or \textit{satisfies} a \textsf{pLTL} formula $\varphi$ at time $t \in \mathbb{N}$, we write $(w, t) \models \varphi$. The semantics is defined in Figure~\ref{fig:ltl semantics}.
\begin{figure}[bt]
\begin{equation*}
\begin{aligned}
    &(w, t) \models \top \\
    &(w, t) \not\models \bot \\
    &\begin{aligned}
    (w, t) &\models p &&\text{ iff } \quad p \in \sigma_{t} \\
    (w, t) &\models \neg \varphi &&\text{ iff } \quad (w, t) \not\models \varphi \\
    (w, t) &\models \psi \land \xi &&\text{ iff } \quad (w, t) \models \psi \text{ and } (w, t) \models \xi \\
    (w, t) &\models \psi \lor \xi &&\text{ iff } \quad (w, t) \models \psi \text{ or } (w, t) \models \xi \\
    (w, t) &\models \ltlX \psi  &&\text{ iff } \quad (w, t + 1) \models \psi \\
    (w, t) &\models \psi \ltlU \xi &&\text{ iff } \quad \exists r \geq t \, . \, ((w, r) \models \xi \text{ and } \forall s \in [t, r) \, . \, (w, s) \models \psi) \\
    (w, t) &\models \ltlY \psi  &&\text{ iff } \quad t > 0 \text{ and } (w, t - 1) \models \psi \\
    (w, t) &\models \psi \ltlS \xi &&\text{ iff } \quad \exists r \in [0..t] \, . \, ((w, r) \models \xi \text{ and } \forall s \in (r, t] \, . \, (w, s) \models \psi).
    \end{aligned}
\end{aligned}
\end{equation*}
%\vspace*{-18pt}
\caption{Semantics of \textsf{pLTL}.\label{fig:ltl semantics}}
%\vspace*{-6pt}
\end{figure}
%\vspace*{-5mm}
%\caption{\label{figure_pltl_semantics}Semantics of \textsf{pLTL}.}
%\vspace*{-7mm}
When $t = 0$ we write $w \models \varphi$. Using the same notation as for the language of automata in Section~\ref{section_preliminaries_automata}, we define the \textit{language} of a \textsf{pLTL} formula $\varphi$ as $\mathcal{L}(\varphi) \coloneqq \{ w \in (\mathcal{P}(AP))^\omega \mid w \models \varphi \}$. Two \textsf{pLTL} formulae $\varphi$ and $\psi$ with the same language are (semantically) \textit{initially equivalent}. In such cases we write $\varphi \equiv_i \psi$.

Additional operators are in Figure~\ref{fig:additional operators}.
\begin{figure}[bt]
\begin{equation*}
\begin{aligned}
    \ltlF \varphi &\coloneqq \top \ltlU \varphi 
    & 
    \ltlG \varphi &\coloneqq \neg \ltlF \neg \varphi 
    & 
    \varphi \ltlW \psi &\coloneqq \varphi \ltlU \psi \lor \ltlG \varphi 
    \\
    \ltlO \varphi &\coloneqq \top \ltlS \varphi 
    &
    \ltlH \varphi &\coloneqq \neg \ltlO \neg \varphi 
    &
    \varphi \ltlWS \psi &\coloneqq \varphi \ltlS \psi \lor \ltlH \varphi 
    \\
    \varphi \ltlM \psi &\coloneqq \psi \ltlU (\varphi \land \psi) 
    & 
    \varphi \ltlR \psi &\coloneqq \psi \ltlW (\varphi \land \psi) 
    &
    \ltlWY \varphi &\coloneqq \ltlY \varphi \lor \neg \ltlY \top
    \\ 
    \varphi \ltlNWS \psi &\coloneqq \psi \ltlS (\varphi \land \psi)
    &
    \varphi \ltlNS \psi &\coloneqq \psi \ltlWS (\varphi \land \psi).
\end{aligned}
\end{equation*}
%\vspace*{-18pt}
\caption{Additional \textsf{pLTL} operators.\label{fig:additional operators}}
%\vspace*{-18pt}
\end{figure}
Notice that $\ltlW$ ($\ltlR$) is the weak version of $\ltlU$ ($\ltlM$).
A formula consisting only of a (possibly negated) proposition is \textit{atomic}. 
A formula whose syntax tree is rooted with $\ltlX, \ltlU, \ltlF, \ltlG, \ltlW, \ltlM$, or $\ltlR$ is a \textit{future} formula, while one rooted with $\ltlY, \ltlS, \ltlO, \ltlH, \ltlWS, \ltlNWS, \ltlNS$, or $\ltlWY$ is a \textit{past} formula.

A \textit{temporal} formula is a formula that is either a future or past formula. We use corresponding terminology for the operators. A subformula of a formula $\varphi$ is \textit{maximal} if it does not appear under temporal operators in the syntax tree of $\varphi$. We write $\Var(\varphi)$ to denote the set of propositional variables appearing in a formula $\varphi$, $\sff(\varphi)$ to denote the set of atomic and temporal subformulae of $\varphi$, and $\psf(\varphi)\subseteq \sff(\varphi)$ the set of past subformulae of $\varphi$.
The \textit{size} $\vert \varphi \vert$ of a formula $\varphi$ is defined as the number of distinct subformulae of $\varphi$ that are either atomic or temporal. That is, we define $\vert \varphi \vert \coloneqq \vert \sff(\varphi) \vert$.

\iffullversion
A formula is in negation normal form if negations only appear in front of atomic propositions. An arbitrary \textsf{pLTL} formula is easily rewritten into an equivalent formula in negation normal form with a linear size increase; see Appendix~\ref{appendix:nnf}. From now on, whenever we write ``formula'', we implicitly refer to \textsf{pLTL} formulae in negation normal form, unless otherwise stated.
\else
A formula is in negation normal form if negations only appear in front of atomic propositions. An arbitrary \textsf{pLTL} formula is easily rewritten into an equivalent formula in negation normal form with a linear size increase. From now on, whenever we write ``formula'', we implicitly refer to \textsf{pLTL} formulae in negation normal form, unless otherwise stated.
\fi

%A formula is in negation normal form if negations only appear in front of atomic propositions. Using the derived operators above it is easy to rewrite an arbitrary \textsf{pLTL} formula into negation normal form without a size increase. From now on, whenever we write ``formula'', we implicitly refer to \textsf{pLTL} formulae in negation normal form, unless otherwise stated.

\subsection{Propositional Equivalence of Formulae}
To obtain reasonable bounds on our automata we need to define an equivalence relation on \textsf{pLTL} formulae. Because semantic equivalence is difficult to determine, we use the stronger notion of \textit{propositional equivalence}. This is equivalent to viewing a formula as a propositional logic formula where maximal temporal subformulae are taken as propositional variables.
\begin{definition}[Propositional Semantics of $\mathsf{pLTL}$~\cite{esparza_unified_translation}]
Let $\mathcal{I}$ be a set of formulae and $\varphi$ a formula. The propositional satisfaction relation $\mathcal{I} \models_p \varphi$ is inductively defined as
\begin{equation*}
\begin{aligned}
    \mathcal{I} &\models_p \top \quad &
    \mathcal{I} &\models_p \psi \land \xi \mbox{ iff } \mathcal{I} \models_p \psi \mbox{ and } \mathcal{I} \models_p \xi \\
\end{aligned}
\end{equation*}
%and
\begin{equation*}
\begin{aligned}
    \mathcal{I} &\not\models_p \bot &
    \mathcal{I} &\models_p \psi \lor \xi \mbox{ iff } \mathcal{I} \models_p \psi \mbox{ or } \mathcal{I} \models_p \xi,
\end{aligned}
\end{equation*}
with $\mathcal{I} \models_p \varphi \mbox{ iff } \varphi \in \mathcal{I}$ for all other cases. Two formulae $\varphi$ and $\psi$ are propositionally equivalent, denoted $\varphi \sim \psi$, if $\mathcal{I} \models_p \varphi \Leftrightarrow \mathcal{I} \models_p \psi$ for all sets of formulae $\mathcal{I}$. The \textit{(propositional) equivalence class} of a formula $\varphi$ is denoted $[\varphi]_\sim$. The \textit{(propositional) quotient set} of a set of formulae $\Psi$ is denoted $\Psi_{/\sim}$.
\end{definition}
We state formally that propositional equivalence implies equivalence.
\begin{lemma}[\cite{esparza_unified_translation}]
    Let $\varphi, \psi$ be formulae such that $\varphi \sim \psi$ , $w$ a word, $t\in \mathbb{N}$. Then, $(w,t)\models \varphi$ iff $(w,t)\models \psi$. 
\end{lemma}
For example, $(a \ltlU b \wedge \ltlG a) \vee  a \ltlU b$ is propositionally equivalent to $a\ltlU b$.  
On the other hand, the formulas $a \ltlU b \vee (a \ltlW b \wedge \ltlF b)$ is equivalent to $a\ltlU b$ but not propositionally equivalent to it. 

The following lemma allows us to lift certain functions from formulae to equivalence classes of formulae.
\begin{lemma}[\cite{esparza_unified_translation}]
\label{lemma_propequiv_preserved}
    Let $f$ be a function on formulae such that $f(\top) = \top$, $f(\bot) = \bot$, and for all formulae $\varphi$ and $\psi$, $f(\varphi \land \psi) = f(\varphi) \land f(\psi)$ and $f(\varphi \lor \psi) = f(\varphi) \lor f(\psi)$. Then, for all pairs of formulae $\varphi$ and $\psi$, if $\varphi \sim \psi$ then $f(\varphi) \sim f(\psi)$.
\end{lemma}

\section{Encoding the Past and the After-Function}
\label{section_previous_work}
In this section we present some definitions and results derived from Azzopardi et al.~\cite{azzopardi_et_al:LIPIcs.MFCS.2024.13} that are needed in order to understand the upcoming translations.
\subsection{Encoding the Past}
\label{section_encoding_past}
Our translations involve taking a formula $\varphi$ \textendash{} essentially representing a state in the automaton \textendash{} and a letter $\sigma$, and from these construct a successor formula by rewriting the formula using the local semantics of \textsf{pLTL}. To accomplish this, we first discuss how to update $\varphi$ so that the information contained in $\sigma$ that is relevant for evaluating the past subformulae of $\varphi$ is encoded. This is done by either weakening or strengthening past operators, in the following sense:
\begin{definition}[Weakening and strengthening formulae]
The \textit{weakening} $\varphi_\mathcal{W}$ and strengthening $\varphi_\mathcal{S}$ of a formula $\varphi$ is defined by case distinction on $\varphi$ as
\begin{equation*}
\begin{aligned}  
    (\ltlY \psi)_\mathcal{W} &\coloneqq \ltlWY \psi & (\psi \ltlS \xi)_\mathcal{W} &\coloneqq \psi \ltlWS \xi & (\psi \ltlNWS \xi)_\mathcal{W} &\coloneqq \psi \ltlNS \xi \\
    (\ltlWY \psi)_\mathcal{S} &\coloneqq \ltlY \psi & (\psi \ltlWS \xi)_\mathcal{S} &\coloneqq \psi \ltlS \xi & (\psi \ltlNS \xi)_\mathcal{S} &\coloneqq \psi \ltlNWS \xi.
\end{aligned}
\end{equation*}
For all other cases we define $\varphi_\mathcal{W} \coloneqq \varphi$ and $\varphi_\mathcal{S} \coloneqq \varphi$.
\end{definition}
To provide some intuition behind this definition, consider a state in an automaton labeled by the formula $\ltlX (p \ltlS q)$, and a word $w$. Upon reading the first letter of $w$, the successor state should \textit{not} be labeled by $p \ltlS q$, as this amounts to making the \textendash{} too strong \textendash{} requirement that $w_1 \models q$ (recall that $w_1$ is the suffix of $w$ without its initial letter). To find the correct successor state, we must also allow for the possibility that $w \models q$ and $w_1 \models p$. In general, we need to be able to encode which past subformulae were true in the \textit{previous} step; in this case, the subformula $p \ltlS q$. Weakening or strengthening past subformulae allows us to do this.

If we know that $w \models q$, then it is sufficient to check that $w_1 \models p \ltlWS q$ to deduce that $(w, 1) \models p \ltlS q$. The formula $q$ is called the \textit{weakening condition} of $p \ltlS q$; its satisfaction allows us to weaken $p \ltlS q$ for the next step.
%Consider the formula $\ltlG\ltlF(\bigvee_{p\in \mathit{AP}}\ltlO p)$. We see that once we read a letter containing some $p\in AP$, the relevant $\ltlO p\equiv_i \top \ltlS p$ is weakened to $\top \ltlWS \varphi$, which is equivalent to $\top$. Thus, the automaton has two states corresponding to the original formula and to $\ltlG\ltlF\top\equiv_i \top$.  
\begin{definition}[Weakening conditions]
Given a past formula $\varphi$, the weakening condition $\wc(\varphi)$ is defined as
\begin{equation*}
\begin{aligned}
\wc(\ltlY \psi) &\coloneqq \psi & \wc(\psi \ltlS \xi) &\coloneqq \xi & \wc(\psi \ltlNWS \xi) &\coloneqq \psi \land \xi \\
\wc(\ltlWY \psi) &\coloneqq \psi & \wc(\psi \ltlWS \xi) &\coloneqq \psi \lor \xi & \wc(\psi \ltlNS \xi) &\coloneqq \xi.
\end{aligned}
\end{equation*}
\end{definition}
We note that formulae rooted with one of the binary operators $\ltlS, \ltlWS, \ltlNWS$, and $\ltlNS$ are initially equivalent to their corresponding weakening conditions:
\begin{proposition}
\label{proposition:wc_equivalent}
    Let $\varphi \op \psi$ be a past formula, where $\op \in \{ \ltlS, \ltlWS, \ltlNWS, \ltlNS \}$. Then $\varphi \equiv_i \wc(\varphi)$.
\end{proposition}
\begin{proof}
    This follows immediately from the semantics of \textsf{pLTL}.
\end{proof}

The following proposition formalizes the intuition given above for how weakening works.
However, to remove the restriction on occurrences of past in subformulae, it needs a recursive application of weakening as done below in Definition~\ref{definition_rewriting_by_set}, leading to the full correctness in Lemma~\ref{lemma:entailed_correct}.
\begin{restatable}{proposition}{wcproperty}
\label{lemma:wc_property}
    Let $\varphi$ be a past formula such that $\psf(\varphi) = \{ \varphi \}$, and $w$ a word. Then $(w, 1) \models \varphi$ if and only if either 1) $w \models \wc(\varphi)$ and $w_1 \models \varphi_\mathcal{W}$, or 2) $w_1 \models \varphi_\mathcal{S}$.
\end{restatable}
We now define how to rewrite a formula $\varphi$ according to a set of past formulae $C$. This set determines which past subformulae of $\varphi$ should be weakened, and which should be strengthened. In general, a past subformula may contain nested future subformulae, and as such may require further input to be evaluated. Hence a set $C$ is analogous to a ``guess'' of precisely which past subformulaes' weakening conditions hold at a given time, and thus should be weakened in preparation for the next input symbol. In the after-function, defined in Section~\ref{section_after}, such guesses are verified by adding the corresponding weakening conditions.
\begin{definition}[Rewriting past operators under sets]
\label{definition_rewriting_by_set}
Let $\varphi$ be a formula and $C$ a set of past formulae. We inductively define $\varphi\rw{C}$ by case distinction on the structure of $\varphi$ as:
\begin{equation*}
\begin{aligned}    
    \psi \rw{C} &\coloneqq \psi && \mbox{($\psi$ atomic)} \\
    (\op \psi)\rw{C} &\coloneqq 
    \begin{cases}
    \begin{aligned}
        &(\op \psi\rw{C})_\mathcal{W} && (\op \psi \in C) \\
        &(\op \psi\rw{C})_\mathcal{S} && \mbox{(otherwise)}
    \end{aligned}
    \end{cases}
    && \mbox{($\op$ unary)} \\
    (\psi \op \xi)\rw{C} &\coloneqq 
    \begin{cases}
    \begin{aligned}
        &(\psi\rw{C} \op \xi\rw{C})_\mathcal{W} && (\psi \op \xi \in C) \\
        &(\psi\rw{C} \op \xi\rw{C})_\mathcal{S} && \mbox{(otherwise)}.
    \end{aligned}
    \end{cases}
    && \mbox{($\op$ binary)}.
\end{aligned}
\end{equation*}
\end{definition}
The following Lemma states the correctness of the rewriting:
\begin{lemma}[\cite{azzopardi_et_al:LIPIcs.MFCS.2024.13}]
\label{lemma:entailed_correct}
  Let $\varphi$ be a formula, $w$ a word, $t \in \mathbb{N}$, and $C = \{ \psi \in \psf(\varphi) \mid w \models \wc(\psi) \}$. Then $(w, t + 1) \models \varphi$ if and only if $(w_1, t) \models \varphi\rw{C}$.
\end{lemma}
Observe that the set $C$ of Lemma~\ref{lemma:entailed_correct} contains exactly the past subformulae of $\varphi$ whose weakening conditions are satisfied by $w$. Hence, a past subformula of $\varphi$ appears in weakened form in $\varphi\rw{C}$ if and only if its weakening condition is satisfied. We call this set the set of past subformulae \textit{entailed} by $w$.
\begin{restatable}{definition}{entailedset}
%\begin{definition}[Sets of entailed subformulae]
\label{definition_set_entailed}
Let $\varphi$ be a formula, $w$ a word, and $t \in \mathbb{N}$. The set of past subformulae of $\varphi$ entailed by $w$ at $t$ is defined as
\begin{equation*}
\begin{aligned}    
    \entconw{\varphi}{0}{w} &\coloneqq \{ \psi \in \psf(\varphi) \mid \psi = \psi_\mathcal{W} \} \\
    \entconw{\varphi}{t}{w} &\coloneqq \{ \psi \in \psf(\varphi\rw{C^w_{\varphi, t - 1}}) \mid w_{t-1} \models \wc(\psi) \} \quad (t > 0).
\end{aligned}
\end{equation*}
\end{restatable}
When the word $w$ above is clear from the context, we simply write $\entcon{\varphi}{t}$. Note that $\entcon{\varphi}{0}$ has no effect on $\varphi$, in the sense that $\varphi\rw{\entcon{\varphi}{0}} = \varphi$; it simply keeps strong past subformulae strong and weak subformulae weak. This is to avoid defining special cases for $t = 0$ when stating properties about sets of entailed subformulae. Given $t \in \mathbb{N}$ we denote the sequence $\entcon{\varphi}{0}, \dots, \entcon{\varphi}{t}$ of length $t + 1$ by $\entseq{\varphi}{t}$ and define $\compseq{\varphi}{t} = \{ \psi \in \psf(\varphi) \mid \psi\rw{\entcon{\varphi}{0}}\rw{\entcon{\varphi}{1}}\dots\rw{\entcon{\varphi}{t}} = (\psi\rw{\entcon{\varphi}{0}}\rw{\entcon{\varphi}{1}}\dots\rw{\entcon{\varphi}{t}})_\mathcal{W.} \}$. In effect, the set $\compseqw{\varphi}{t}{w}$ immediately puts $\psi$ in the correct ``past configuration'' at time $t$.
\subsection{The After-Function}
\label{section_after}
To compute a (deterministic) transition from a state \textendash{} represented by a formula \textendash{} given a letter, we make use of the local semantics of \textsf{pLTL} and the results of Section~\ref{section_encoding_past} to update the formula, intuitively returning its ``derivative'' with respect to a given letter. We first define the \textit{local after-function}.
This is essentially the ``classic'' definition of the after function, which does not work for past formulas.
It is completed by the (global) \textit{after-function} below. 
\begin{definition}[The local after-function]
Let $\varphi$ be a formula and $\sigma \in \mathcal{P}(\Var(\varphi))$. The local after-function $\afc[]$ is inductively defined on the structure of $\varphi$ as:
\begin{align*}
    \afl(p) &\coloneqq
    \begin{cases}
        \top &(p \in \sigma) \\
        \bot &(p \not\in \sigma)
    \end{cases} &
    \afl(\neg p) &\coloneqq 
    \begin{cases}
        \bot &(p \in \sigma) \\
        \top &(p \not\in \sigma)
    \end{cases} & \\
    \afl(\psi \land \xi) &\coloneqq \afl(\psi) \land \afl(\xi) &
    \afl(\psi \lor \xi) &\coloneqq \afl(\psi) \lor \afl(\xi) \\
    \afl(\ltlX \psi) &\coloneqq \psi & \afl(\ltlY \psi) &\coloneqq \bot \\
    \afl(\psi \ltlU \xi) &\coloneqq \afl(\xi) \lor \afl(\psi) \land \psi \ltlU \xi & \afl(\psi \ltlS \xi) &\coloneqq \afl(\xi) \\
    \afl(\psi \ltlW \xi) &\coloneqq \afl(\xi) \lor \afl(\psi) \land \psi \ltlW \xi & \afl(\psi \ltlWS \xi) &\coloneqq \afl(\psi \lor \xi) \\
    \afl(\psi \ltlM \xi) &\coloneqq \afl(\xi) \land \big( \afl(\psi) \lor \psi \ltlM \xi \big) & \afl(\psi \ltlNWS \xi) &\coloneqq \afl(\psi \land \xi) \\
    \afl(\psi \ltlR \xi) &\coloneqq \afl(\xi) \land \big( \afl(\psi) \lor \psi \ltlR \xi \big) & \afl(\psi \ltlNS \xi) &\coloneqq \afl(\xi) \\
    &&\afl(\ltlWY \psi) &\coloneqq \top.
\end{align*}
\end{definition}
From the above definition we see that subformulae that can be fully evaluated using only the information in $\sigma$ are immediately replaced by $\top$ or $\bot$, while subformulae that require further input (those involving $\ltlX, \ltlU, \ltlW, \ltlM$, and $\ltlR$) are partially evaluated. Because formulae rooted with binary past operators are equivalent to their weakening conditions (see Proposition~\ref{proposition:wc_equivalent}), the latter are immediately evaluated in place of the former.
For example, consider the formulas $\varphi_1:= (a \ltlS b) U c$ and $\varphi_2:= (a \ltlWS b) U c$. Computing the after of $\varphi_1$ reading $\{b\}$ leads to $\afl[\{b\}](c) \vee (\afl[\{b\}](a \ltlS b) \wedge \varphi_1)$, which is $\bot \vee (\varphi_1 \wedge \top)\equiv \varphi_1$. Indeed, $\afl[\{b\}](c)\equiv\bot$ while $\afl[\{b\}](a\ltlS b)\equiv \afl[\{b\}]{b}\equiv \top$. 
Computing the after of $\varphi_1$ reading $\{a\}$ we see that it evaluates to $\bot$. 
Contrarily, the after of $\varphi_2$ reading $\{a\}$ evaluates to $\varphi_2$. Indeed, $\afl[\{a\}](a \ltlWS b)\equiv \top$. 
Notice, that this does not include the update of past information, which is included below. 

Next, we define the (global) after-function. This applies the local after-function, and then applies the rewrite $\langle C \rangle$ for each $C \subseteq \psf(\varphi)$. 
The rewrite essentially naturally embeds the memory of the past in the structure of the formula by replacing strong to weak past operators.
The disjunction over all sets ensures that the ``correct'' set is among those considered, so that the definition is complete with respect to the semantics of \textsf{pLTL}. Moreover, for each $C$ the weakening conditions associated with the formulae in $C$ are added to ensure that $C$ does not weaken ``too much''.
\begin{definition}[The after-function]
\label{definition_after_function}
Let $\varphi$ be a formula and $\sigma$ a letter. 
The after-function $\af[]{}$ is defined as:
\begin{equation*}
    \af(\varphi) \coloneqq \bigvee_{C \subseteq \psf(\varphi)} \Big( \afl(\varphi) \land \bigwedge_{\psi \in C} \afl\big( \wc(\psi) \big) \Big)
    \rw{C},
\end{equation*}
On occasion we write $\af(\varphi, C)$ to refer to a specific disjunct of $\af(\varphi)$, i.e. we let $\af(\varphi, C) \coloneqq \big(\afl(\varphi) \land \bigwedge_{\psi \in C} \afl(\wc(\psi))\big)\rw{C}$.
\end{definition}
Revisiting $\varphi_1 := (a \ltlS b) U c$, and computing $\af[\{b\}](\varphi_1)$ we get a disjunction between $(\afl[\{b\}](\varphi_1)\wedge \afl[\{b\}](\wc(a \ltlS b)))\rw{\{a \ltlS b\}}$ and $(\afl[\{b\}](\varphi_1)\wedge \top)\rw{\emptyset}$.
The former includes $\wc(a \ltlS b)\equiv \afl[\{b\}](b)$ as a conjunct to ensure that the rewrite by $\{a \ltlS b\}$ is correct. The resulting former disjunct is equivalent to $\varphi_2$ and the latter to $\varphi_1$. As $\varphi_1\rightarrow \varphi_2$, the disjunction can be simplified to $\varphi_2$. This internalizes naturally that $b$, as the consequent of $a\ltlS b$, holds and turns it to $a\ltlWS b$.
Now, $\af[\{a\}](\varphi_2)$ is $\varphi_2$ again, correctly 
maintaining the continuous holding of $a \ltlS b$ after reading $\{b\}\cdot \{a\}$. 
\begin{remark}
\label{remark:optimization}
    As an optimization, observe that replacing each maximal past subformula in $\af(\varphi)$ of the form $\ltlY \psi$ by $\bot$ and each maximal past subformula of the form $\ltlWY \psi$ by $\top$ results in an equivalent formula. Similarly, it follows from Proposition~\ref{proposition:wc_equivalent} that each maximal past subformula rooted with either $\ltlS, \ltlWS, \ltlNWS$, or $\ltlWS$, can be replaced by their respective weakening conditions. By doing this until no maximal past subformulae remain we obtain a formula equivalent to $\af(\varphi)$ such that all of its maximal formulae are future formulae. We omit this optimization in order to avoid further complicating the presentation.
\end{remark}

The addition of weakening conditions corresponding to each subset $C \subseteq \psf(\varphi)$ ensures that the definition above is also sound, leading to the main results of this section:
\begin{restatable}{lemma}{entailedcorrectb}
\label{lemma:entailed_correctb}
    Let $\varphi$ be a formula, $\sigma$ a letter, and $w$ a word starting with $\sigma$. Then
    \begin{itemize}
        \item If $w \models \varphi$, then $w_1 \models \af(\varphi, C)$, where $C = \{ \psi \in \psf(\varphi) \mid w \models \wc(\psi) \}$.
        \item If $w_1 \models \af(\varphi, C)$, where $C \subseteq \psf(\varphi)$ is arbitrary, then $w \models \varphi$.
    \end{itemize}
\end{restatable}
From this we obtain the following corollary:
\begin{restatable}{corollary}{afcorrect}
\label{lemma:af_correct}
    Let $\varphi$ be a formula, $\sigma$ a letter, and $w$ a word starting with $\sigma$. Then $w \models \varphi$ iff $w_1 \models \af(\varphi)$.
\end{restatable}
The extension of $\af[]$ to finite words is done in the natural way: given a formula $\varphi$ and word $w$ of length $n$, we define
\begin{equation*}
\begin{aligned}    
    \af[\epsilon](\varphi) &\coloneqq \varphi &&(n = 0) \\
    \af[w_{0n}](\varphi) &\coloneqq \af[w_{(n-1)n}](\af[w_{0(n-1)}](\varphi)) &&(n > 0).
\end{aligned}
\end{equation*}
Observe that Corollary~\ref{lemma:af_correct} immediately extends to finite words by induction.

\section{From \texorpdfstring{\textsf{pLTL}}{pLTL} To LDBA}
\label{section_pltl_to_ldba}
\subsection{The Limit Behavior of Words}
\label{section_limit_behavior}
Here we briefly outline the results by Esparza et al.~\cite{esparza_unified_translation}, adapted to the setting of \textsf{pLTL} by Azzopardi et al.~\cite{azzopardi_et_al:LIPIcs.MFCS.2024.13}, that are of relevance for this construction. For a more thorough exposition of these concepts we refer to the former for their introduction and to the latter for their adaptation to \textsf{pLTL}.

Let $\mu$-$\mathsf{pLTL}$ denote the set of formulae whose future operators are members of $\{ \ltlX, \ltlU, \ltlM \}$, and $\nu$-$\mathsf{pLTL}$ the set of formulae whose future operators are members of $\{ \ltlX, \ltlW, \ltlR \}$. Given a formula $\varphi$, we denote by $\mu(\varphi)$ the set of subformulae of $\varphi$ whose syntax trees are rooted with $\ltlU$ or $\ltlM$ and by $\nu(\varphi)$ the set of subformulae of $\varphi$ whose syntax trees are rooted with $\ltlW$ or $\ltlR$. The relevance of these definitions is that we can easily obtain automata equivalent to formulae in $\mu$-\textsf{pLTL} and $\nu$-\textsf{pLTL} using the after-function: given a formula in $\mu$-\textsf{pLTL} we can construct an equivalent automaton by taking propositional equivalence classes of formulae as states, using the after-function lifted to such classes as transition function, and defining the set $[\top]_\sim$ as accepting. For a formula $\varphi \in \nu$-\textsf{pLTL} we instead take all states as accepting except $[\bot]_\sim$. This is motivated by the following lemma:
\begin{lemma}[\cite{azzopardi_et_al:LIPIcs.MFCS.2024.13}]
\label{lemma_mu_nu_limits}
    Let $\varphi$ be a formula and $w$ a word. Then
    \begin{equation*}
    \begin{aligned}
        &1) \mbox{ If } \varphi \in \mu\mbox{-}\mathsf{pLTL} \mbox{ then } w \models \varphi \Leftrightarrow \exists t \, . \, \af[w_{0t}](\varphi) \sim \top. \\
        &2) \mbox{ If } \varphi \in \nu\mbox{-}\mathsf{pLTL} \mbox{ then } w \models \varphi \Leftrightarrow \forall t \, . \, \af[w_{0t}](\varphi) \not\sim \bot.
    \end{aligned}
    \end{equation*}
\end{lemma}
Consider a formula $\varphi$, a word $w$, and two disjoint subsets $\langle \Phi, \Psi \rangle$ of subformulae of $\varphi$. It is easy to see that there exists some index after which all subformulae in $\Phi$ are satisfied by $w$ infinitely often or not at all, and all subformulae in $\Psi$ are satisfied by $w$ either persistently (i.e., from that point onward) or are false infinitely often. In particular, this is true in the special case when $\Phi = \mu(\varphi)$ and $\Psi = \nu(\varphi)$. We are interested in finding subsets $M \subseteq \mu(\varphi)$ and $N \subseteq \nu(\varphi)$ that precisely capture this ``limit behavior'' of $w$, in the sense that there exists a $t$ such that $M$ coincides with the sets $\F{\varphi}{w}{t}$ and $\GF{\varphi}{w}{t}$ below, and $N$ with the sets $\G{\varphi}{w}{t}$ and $\FG{\varphi}{w}{t}$:
\begin{equation*}
\begin{aligned}
    &\F{\varphi}{w}{t} \coloneqq \{ \psi \in \mu(\varphi) \mid (w, t) \models \ltlF \psi \}
    &&\GF{\varphi}{w}{t} \coloneqq \{ \psi \in \mu(\varphi) \mid (w, t) \models \ltlG \ltlF \psi \} \\
    &\G{\varphi}{w}{t} \coloneqq \{ \psi \in \nu(\varphi) \mid (w, t) \models \ltlG \psi \}
    &&\FG{\varphi}{w}{t} \coloneqq \{ \psi \in \nu(\varphi) \mid (w, t) \models \ltlF \ltlG \psi \}.
\end{aligned}
\end{equation*}
\begin{definition}[Stable words]
    A word $w$ is $\mu$-\textit{stable} ($\nu$-\textit{stable}) with respect to a formula $\varphi$ at index $t$ if $\F{\varphi}{w}{t} = \GF{\varphi}{w}{t}$ ($\G{\varphi}{w}{t} = \FG{\varphi}{w}{t}$). If $w$ is both $\mu$-stable and $\nu$-stable with respect to $\varphi$ at index $t$, then it is \textit{stable} with respect to $\varphi$ at %index 
    $t$.
\end{definition}
Assume that $w$ is stable with respect to $\varphi$ at $t$. If we know the limit behavior of $w$ with respect to $\varphi$, in the sense that we know either which formulae in $\mu(\varphi)$ hold infinitely often or which formulae in $\nu(\varphi)$ hold persistently, then at $t$ we can rewrite $\varphi$ into either an equivalent $\nu$-\textsf{pLTL} formula or an equivalent $\mu$-\textsf{pLTL} formula, respectively. These rewrites are done according to definitions \ref{definition_M} and \ref{definition_MN_transform}:
\begin{definition}
\label{definition_M}
Let $\varphi$ be a formula and $M$ a set of formulae whose syntax trees are rooted with $\ltlU$ or $\ltlM$. The formula $\varphi[M]_\nu$ is inductively defined as
\begin{equation*}
\begin{aligned}
    \psi[M]_\nu &\coloneqq \psi &&(\psi \, \text{atomic}) \\
    (\op \psi)[M]_\nu &\coloneqq \op (\psi[M]_\nu) &&(\op \text{unary}) \\
    (\psi \op \xi)[M]_\nu &\coloneqq (\psi[M]_\nu) \op (\xi[M]_\nu) &&(\op \in \{ \land, \lor, \ltlW, \ltlR, \ltlS, \ltlWS, \ltlNWS, \ltlNS \}) \\
    (\psi \ltlU \xi)[M]_\nu &\coloneqq
    \mathrlap{\begin{cases}
        (\psi[M]_\nu) \ltlW (\xi[M]_\nu) &(\psi \ltlU \xi \in M) \\
        \bot &(\mbox{otherwise})
    \end{cases}}
    \\
    (\psi \ltlM \xi)[M]_\nu &\coloneqq
    \mathrlap{\begin{cases}
        (\psi[M]_\nu) \ltlR (\xi[M]_\nu) &(\psi \ltlM \xi \in M) \\
        \bot &(\mbox{otherwise}),
    \end{cases}}
\end{aligned}
\end{equation*}
\end{definition}
\begin{definition}
\label{definition_MN_transform}
Let $\varphi$ be a formula and $N$ a set of formulae whose syntax trees are rooted with $\ltlW$ or $\ltlR$. The formula $\varphi[N]_\mu$ is inductively defined as
\begin{equation*}
\begin{aligned}
    (\psi \ltlW \xi)[N]_\mu &\coloneqq
    \begin{cases}
        \top &(\psi \ltlW \xi \in N) \\
        (\psi[N]_\mu) \ltlU (\xi[N]_\mu) &(\mbox{otherwise})
    \end{cases} \\
    (\psi \ltlR \xi)[N]_\mu &\coloneqq
    \begin{cases}
        \top &(\psi \ltlR \xi \in N) \\
        (\psi[N]_\mu) \ltlM (\xi[N]_\mu) &(\mbox{otherwise}).
    \end{cases}
\end{aligned}
\end{equation*}
The other cases are defined by recursive descent similarly to Definition~\ref{definition_M}.
\end{definition}
The conditions under which such rewrites result in equisatisfiable formulae with respect to a given word are given by the following lemma:
\begin{lemma}[\cite{azzopardi_et_al:LIPIcs.MFCS.2024.13}]
\label{lemma_MN_properties}
    Let $\varphi$ be a formula, $w$ a word, and $t \in \mathbb{N}$. Let $M$ and $N$ be sets of formulae. Then
    \begin{equation*}
    \begin{aligned}
        &1) \mbox{ If } \F{\varphi}{w}{0} \subseteq M \mbox{ and } (w, t) \models \varphi, \mbox{then } (w, t) \models \varphi[M]_\nu. \\
        &2) \mbox{ If } M \subseteq \GF{\varphi}{w}{0} \mbox{ and } (w, t) \models \varphi[M]_\nu, \mbox{then } (w, t) \models \varphi. \\
        &3) \mbox{ If } \FG{\varphi}{w}{0} \subseteq N \mbox{ and } (w, t) \models \varphi, \mbox{then } (w, t) \models \varphi[N]_\mu. \\
        &4) \mbox{ If } N \subseteq \G{\varphi}{w}{0} \mbox{ and } (w, t) \models \varphi[N]_\mu, \mbox{then } (w, t) \models \varphi.
    \end{aligned}
    \end{equation*}
\end{lemma}
As stated, Lemma~\ref{lemma_MN_properties} describes under which assumptions rewrites by $\placeholder[M]_\nu$ and $\placeholder[N]_\mu$ result in equisatisfiable formulae with respect to a given word. Of particular interest is that the resulting formulae are in $\nu\mbox{-}\mathsf{pLTL}$ and $\mu\mbox{-}\mathsf{pLTL}$, respectively. By Lemma~\ref{lemma_mu_nu_limits} and the discussion that precedes it, the satisfaction of such formulae is easily checked using Büchi automata. Lemma~\ref{lemma_MN_=>} and Lemma~\ref{lemma_MN_<=} describe how to check the assumptions on $M$ and $N$.
\begin{lemma}[\cite{azzopardi_et_al:LIPIcs.MFCS.2024.13}]
\label{lemma_MN_=>}
    Let $\varphi$ be a formula and $w$ a word that is stable with respect to $\varphi$ at index $s$. If $M = \GF{\varphi}{w}{0}$ and $N = \FG{\varphi}{w}{0}$, then
    \begin{gather*}
        \forall \psi \in M \, . \, \forall t \geq 0 \, . \, w_t \models \ltlF (\psi\rw{\compseq{\varphi}{t}}[N\rw{\compseq{\varphi}{t}}]_\mu) \\
        \mbox{and} \\
        \forall \psi \in N \, . \, \forall t \geq s \, . \, w_t \models \ltlG (\psi\rw{\compseq{\varphi}{t}}[M\rw{\compseq{\varphi}{t}}]_\nu).
    \end{gather*}
\end{lemma}
\begin{lemma}[\cite{azzopardi_et_al:LIPIcs.MFCS.2024.13}]
\label{lemma_MN_<=}
    Let $\varphi$ be a formula and $w$ a word. Let $M \subseteq \mu(\varphi)$ and $N \subseteq \nu(\varphi)$. If
    \begin{gather*}
        \forall \psi \in M \, . \, \forall r \, . \, \exists t \geq r \, . \, w_t \models \ltlF (\psi\rw{\compseq{\varphi}{t}}[N\rw{\compseq{\varphi}{t}}]_\mu) \\
        \mbox{and} \\
        \forall \psi \in N \, . \, \exists t \geq 0 \, . \, w_t \models \ltlG (\psi\rw{\compseq{\varphi}{t}}[M\rw{\compseq{\varphi}{t}}]_\nu),
    \end{gather*}
    then $M \subseteq \GF{\varphi}{w}{0}$ and $N \subseteq \FG{\varphi}{w}{0}$,
\end{lemma}
\subsection{The Idea Behind the Construction}
\label{section_ldba_overview}
Lemma~\ref{lemma_MN_properties} tells us that we can freely rewrite $\varphi$ by either $\placeholder[M]_\nu$ or $\placeholder[N]_\mu$ to obtain a formula that is equivalent with respect to $w$ at a point of stability, under certain assumptions on $M$ and $N$. Our strategy in constructing an LDBA, then, is to update $\varphi$ by successively applying the after-function and make a guess about when the point of stability has been reached, at which point we also guess which two sets $M \subseteq \mu(\varphi)$ and $N \subseteq \nu(\varphi)$ capture the limit behavior of $w$; this makes up the non-deterministic part of the automaton. After making such a guess, we rewrite $\varphi$ by $\placeholder[M]_\nu$, so that the guess can be verified with an automaton with a Büchi acceptance condition. At this point we must also check that every formula in $M$ is satisfied infinitely often, and that every formula in $N$ is satisfied persistently (i.e., from this point forever).

Throughout this section we consider a fixed formula $\varphi$ that is to be translated to a limit-deterministic Büchi automaton and a fixed ordering $C_1, C_2, \dots, C_{2^{\vert \psf(\varphi) \vert}}$ of the elements of $\mathcal{P}(\psf(\varphi))$. For simplicity, we assume $C_1 = \{ \psi \in \psf(\varphi) \mid \psi = \psi_\mathcal{W} \}$. This ensures that $C_1$ describes the initial configuration of $\varphi$, in the sense that $\psi\rw{C_1} = \psi$ for all $\psi \in \sff(\varphi)$.

The state spaces of the automata we construct are defined in terms of the following notions. Given a formula $\varphi$ we denote by $\mathbb{B}(\varphi)$ the set of formulae $\psi$ satisfying $\sff(\psi) \subseteq \sff(\varphi) \cup \{ \top, \bot \}$. That is, $\mathbb{B}(\varphi)$ is the set of formulae that are Boolean combinations of temporal and atomic subformulae of $\varphi$. Similarly, consider a formula that is a disjunction of formulae with the property that for each such disjunct $\psi$ there exists a $C \in 2^{\psf(\varphi)}$ such that $\psi \in \mathbb{B}(\varphi\rw{C})$. We denote by $\mathbb{B}^\lor(\varphi)$ the set of all such formulae.

That the automata constructed are of asymptotically optimal size is a consequence of the following lemma.
\begin{lemma}[\cite{azzopardi_et_al:LIPIcs.MFCS.2024.13}]
\label{lemma_size_bound}
Let $\varphi$ be a formula. Then $\left\vert \mathbb{B}^\lor(\varphi)_{/\sim} \right\vert \leq 2^{2^{\vert \varphi \vert + \vert \psf(\varphi) \vert}}$.
\end{lemma}
We note that implementing the optimization described in Remark~\ref{remark:optimization} further improves this bound to $\left\vert \mathbb{B}^\lor(\varphi)_{/\sim} \right\vert \leq 2^{2^{\vert \varphi \vert}}$.

\begin{remark}
For the remainder of this section we freely make use of $\af$, $\placeholder[\placeholder]_\nu$, and $\placeholder[\placeholder]_\mu$ lifted to equivalence classes of formulae. This is justified by Lemma~\ref{lemma_propequiv_preserved}. Furthermore, while the states of all automata are equivalence classes of formulae, we omit for the sake of readability explicit use of the notation $[\placeholder]_\sim$ and $\placeholder_{/\sim}$.
\end{remark}

\subsection{The Weakening Conditions Automaton}
\label{section_wc_automaton}
As mentioned in Section~\ref{section_ldba_overview}, given two sets $M \subseteq \mu(\varphi)$ and $N \subseteq \nu(\varphi)$, we need a way to check the limit-behavior of the given word with respect to the formulae in these sets. By rewriting the formulae in $M$ by $\placeholder[N]_\mu$ and the formulae in $N$ by $\placeholder[M]_\nu$, this can be done by automata using Büchi acceptance conditions, according to Lemma~\ref{lemma_mu_nu_limits}. In the case of standard \textsf{}{LTL}, checking that a formula $\psi \in M$ is satisfied infinitely often is simple: once the accepting state $[\top]_\sim$ for the formula $\ltlF (\psi[N]_\mu)$ is reached, the automaton loops back to the initial state to check the formula again. Similarly, it is easy to check that a formula $\psi \in N$ is satisfied persistently, by rejecting if the sink state $[\bot]_\sim$ is reached for the formula $\ltlG (\psi[M]_\nu)$.

The presence of past operators adds an extra complication in the construction of these automata, in that care must be taken to properly encode the past when checking a formula in $M$ or $N$ at a specific point in time. Because the after-function is destructive in the sense that it removes subformulae in which the past is encoded, we need a way to track weakening conditions on the side. To this end, we define the \textit{weakening conditions (WC) automaton} $\mathcal{H}_\varphi = (S, S_0, \delta_\mathcal{H})$. First, we define \textit{saturated} sets of past formulae. 
\iffullversion
This definition is a technicality that allows us to prove the desired bounds on the size of the set of states of our automata, and it is motivated further in Appendix~\ref{appendix_saturated}.
\else
This definition is a technicality that allows us to prove the desired bounds on the size of the set of states of our automata.
\fi
\begin{definition}[Saturated sets of past formulae]
Let $\varphi$ be a formula and let $C$ and  $C'$ be sets of past formulae. Then $C'$ is saturated with respect to $C$, denoted $C \preceq C'$, if for all $\xi, \xi' \in \psf(\varphi)$ such that $\xi\rw{C} = \xi'\rw{C}$, it holds that $\xi\rw{C'} = \xi'\rw{C'}$.
\end{definition}

The WC automaton tracks, for each index $i \in [2^{\vert \psf(\varphi) \vert}]$ and every time step, the weakening conditions that must be satisfied in order to apply the rewrite $\placeholder\rw{C_i}$. This leads to the following definitions of the set of states $S$ and initial state $S_0$:
\begin{align*}
    S &\coloneqq \prod_{i = 1}^{2^{\vert \psf(\varphi) \vert}} \mathbb{B}(\varphi\rw{C_i}) \\
    S_0 &\coloneqq \langle \top, \bot, \dots, \bot \rangle
\end{align*}
Its transition function is defined by $\delta_\mathcal{H}(\langle \psi_1, \psi_2, \dots, \psi_{2^{\vert \psf(\varphi) \vert}}\rangle, \sigma) \coloneqq \langle \psi'_1, \psi'_2, \dots, \psi'_{2^{\vert \psf(\varphi) \vert}}\rangle$, where
\begin{equation*}
    \psi'_i \coloneqq \bigvee_{j \in J_i} \bigg(
    \af(\psi_j, C_i\rw{C_j}) \land
    \bigwedge_{\xi \in C_i}
    \af(\wc(\xi\rw{C_j}), C_i\rw{C_j})
    \bigg),
\end{equation*}
and $J_i \coloneqq \{ j \in [2^{\vert \psf(\varphi) \vert}] \mid C_j \preceq C_i \}$.

This is to be interpreted as follows: the formula $\psi_i'$ is intended to capture what must be satisfied in order to apply the rewrite $\placeholder \rw{C_i}$. The outer disjunction over $J_i$ corresponds to making a guess of which $C_j$ was the correct choice of set to rewrite with in the previous step. Since $C_i$ is a subset of $\psf(\varphi)$, we update it to reflect the current configuration of $\varphi$ by rewriting it by $\placeholder\rw{C_j}$.
Now, having made a guess of current set $C_i$ and previous set $C_j$, we update the existing condition $\psi_j$ that was added to justify the application of $C_j$ in the previous step by adding the formula $\afc(\psi_j, C_i\rw{C_j})$. We must also add new weakening conditions to justify the weakenings made by the rewrite of $C_i$, and update them with the current letter $\sigma$. We do this by adding the conjunct $\bigwedge_{\xi \in C_i} \af(\wc(\xi\rw{C_j}), C_i\rw{C_j})$. Observe that also here each $\xi \in C_i$ is updated by $C_j$ to reflect the current configuration of $\varphi$.
\subsection{The Initialization Formulae}
\label{section_reset_formulae}
With a mechanism in place for tracking the weakening conditions that justify each rewrite, we now define formulae describing the states reached upon resetting the automata that verify guessed sets $M$ and $N$.
\begin{definition}[The initialization formulae]
\label{definition_reset_formulae}
Let $\xi^\times = \langle \xi_1, \xi_2, \dots, \xi_{2^{\vert \psf(\varphi) \vert}} \rangle \in S$ and $\psi \in \sff(\varphi)$. Given a set $M \subseteq \mu(\varphi)$ or $N \subseteq \nu(\varphi)$, we define the initialization formulae by
\begin{align*}
    \iota(\psi, \xi^\times, M) &\coloneqq
    \bigvee_{i \in [2^{\vert \psf(\varphi) \vert}]} \big( \psi\rw{C_i}[M\rw{C_i}]_\nu \land \xi_i[M\rw{C_i}]_\nu \big).
    \\
    \iota(\psi, \xi^\times, N) &\coloneqq
    \bigvee_{i \in [2^{\vert \psf(\varphi) \vert}]} \big( \psi\rw{C_i}[N\rw{C_i}]_\mu \land \xi_i[N\rw{C_i}]_\mu \big).
\end{align*}
\end{definition}
To motivate the definition of $\iota(\psi, \xi^\times, M)$ (the motivation for $\iota(\psi, \xi^\times, N)$ is similar), we reconstruct the formula in steps. Consider that we want to verify that $(w, t) \models \psi$. Assuming that $i$ is the index such that $C_i$ is the set that correctly weakens and strengthens past subformulae of $\varphi$ at time $t$, this is equivalent to checking that $w_t \models \psi\rw{C_i}$. To justify this rewrite, we add the corresponding weakening conditions \textendash{} tracked by the WC automaton \textendash{} resulting in the formula $\psi\rw{C_i} \land \xi_i$. In order to easily check the resulting formula using Lemma~\ref{lemma_mu_nu_limits} and Lemma~\ref{lemma_MN_properties}, we want to rewrite it into a $\nu\mbox{-}\mathsf{pLTL}$-formula. However, directly rewriting it by $M$ to obtain the formula $\psi\rw{C_i}[M]_\nu \land \xi_i[M]_\nu$ is not appropriate; each past subformula $\zeta$ of $\varphi$ appears as $\zeta$ in $M$, but as $\zeta\rw{C_i}$ in $\psi\rw{C_i}$ and $\xi_i$. For this reason, we also rewrite the formulae in $M$ by $C_i$, resulting in the formula $\psi\rw{C_i}[M\rw{C_i}]_\nu \land \xi_i[M{\rw{C_i}}]_\nu$. Finally, as the index $i$ is unknown, we consider the disjunction over all indices.
%To provide some intuition for this definition, consider that we want to verify that $(w, t) \models \psi$. Assuming that $C_i$ is the set that correctly weakens and strengthens past subformulae of $\psi$, this is equivalent to checking that $w_t \models \psi\rw{C_i}$. As the index $i$ is unknown, we consider the disjunction over all indices in $\vert \mathcal{P}(\psf(\varphi)) \vert$. For each such index, the corresponding weakening conditions $\xi_i$, as tracked by the WC automaton, are added. In order to easily check the resulting formula using Lemma~\ref{lemma_mu_nu_limits} and Lemma~\ref{lemma_MN_properties}, we want to rewrite it by either $\placeholder[M]_\nu$ or $\placeholder[N]_\mu$. However, we must first update the formulae in $M$ or $N$ to be consistent with how they appear as subformulae of $\psi\rw{C_i}$ and $\xi_i$, motivating the rewrite $M\rw{C_i}$ and $N\rw{C_i}$.
\subsection{Components of the LDBA}
\label{section_ldba_components}
We begin by defining three components that make up the accepting part of the LDBA. These are multiple-entry runner automata over $\mathcal{P}(\Var(\varphi))$, all with bed automaton $\mathcal{H}_\varphi = (S, S_0, \delta_\mathcal{H})$. The reason for using multiple-entry automata is that the state reached upon making a jump into the accepting part depends on the current state reached in the non-deterministic part.
\subsubsection{Verifying the Main Formula}
\label{section_main_automaton}
We begin by defining the multiple-entry runner automaton $\mathcal{B}^1_{M, N} = (Q^1_{M, N}, \delta^1_{M, N}, \alpha^1_{M, N})$ over $\mathcal{P}(\Var(\varphi))$ with bed automaton $\mathcal{H}_\varphi$ and the following components:
\begin{align*}
    Q^1_{M, N} &\coloneqq \mathbb{B}^\lor(\varphi[M]_\nu)
    &
    \delta^1_{M, N}(\zeta, \xi^\times, \sigma) &\coloneqq \af(\zeta)
    &
    \alpha^1_{M, N} &\coloneqq Q^1_{M, N} \setminus \bot.
\end{align*}
\begin{remark}
    We reiterate that we are omitting notation related to propositional equivalence. For example, it is really the case that the state space $Q^1_{M, N}$ of $\mathcal{B}^1_{M, N}$ is $ \mathbb{B}^\lor(\varphi[M]_\nu)_{/\sim}$, and that its transition function $\delta^1_{M, N}([\zeta]_\sim, \xi^\times, \sigma)$ is $[\af(\zeta, \sigma)]_\sim$.
\end{remark}
Consider a word $w$ being read by an LDBA of the kind described in Section~\ref{section_ldba_overview}, and that the run jumps, after having read $w_{0t}$, from the non-deterministic part to the accepting part described by $\mathcal{H}_\varphi \ltimes \mathcal{B}^1_{M, N}$. At this point $\varphi$ has been transformed into $\af[w_{0t}](\varphi)$. At the same time, the WC automaton is in the state $[\delta_\mathcal{H}(S_0, w_{0t})]_\sim$. The purpose of $\mathcal{B}^1_{M, N}$ is to verify that $w_t \models \af[w_{0t}](\varphi)[M\rw{C}]_\nu$ for some $C$, accepting if and only if $w \models \varphi$ (given that the conditions on $M$ and $C$ are checked elsewhere).
\begin{restatable}{lemma}{lemmabone}
\label{lemma_B1}
    Let $w$ be a word over $\mathcal{P}(\Var(\varphi))$, $t \in \mathbb{N}$, and $\eta$ the run of $\mathcal{H}_\varphi \ltimes \mathcal{B}^1_{M, N}$ on $w_t$ starting from the state $\langle \beta, \gamma \rangle$, where $\beta = \delta_\mathcal{H}(S_0, w_{0t})$ and $\gamma = \iota(\af[w_{0t}](\varphi), \beta, M)$. Then
    \begin{enumerate}
    \item If $\GF{\varphi}{w}{0} \subseteq M$, $w \models \varphi$, and $w$ is stable at $t$, then $\eta$ is an accepting run.
    \item If $M \subseteq \GF{\varphi}{w}{0}$ and $\eta$ is an accepting run then $w \models \varphi$.
    \end{enumerate}
\end{restatable}

\subsubsection{Verifying the Choice of \texorpdfstring{$M$}{M}}
\label{section_M_automaton}
We define a multipe-entry runner automaton $\mathcal{B}^2_{M, N}$ that verifies the correctness of choice of $M$. 

Let $\psi_1, \psi_2, \dots, \psi_{\vert M \vert}$ be an ordering of the formulae in $M$. We now define the multiple-entry runner automaton $\mathcal{B}^2_{M, N} = (Q^2_{M, N}, \delta^2_{M, N}, \alpha^2_{M, N})$ over $\mathcal{P}(\Var(\varphi))$ with bed automaton $\mathcal{H}_\varphi$ by
{\allowdisplaybreaks
\begin{equation*}
\begin{aligned}
    Q^2_{M, N} &\coloneqq \bigcup_{i \in [\vert M \vert]} \big\{\langle \xi, i \rangle \mid \xi \in \mathbb{B}^\lor(\varphi[N]_\mu \land \ltlF(\psi_i[N]_\mu)) \big\} \\
    \delta^2_{M, N}\big(\big\langle \zeta, i \big\rangle, \xi^\times, \sigma\big) &\coloneqq
    \begin{cases}
    \begin{aligned}
        &\big\langle \af(\zeta), i \big\rangle &&(\zeta \not\sim \top) \\
        &\big\langle
        \iota(\ltlF \psi_j, \xi^\times, N)
        , j
        \big\rangle &&(\zeta \sim \top), \\
    \end{aligned}
    \end{cases} \\
    &\quad \quad \quad \text{where } j = (i\bmod{\vert M \vert}) + 1 \\
    \alpha^2_{M, N} &\coloneqq \big\{ \big\langle \top, 1 \big\rangle \big\}.
\end{aligned}
\end{equation*}
}
We extend propositional equivalence to states of $Q^2_{M, N}$ in the obvious way: two states $\langle [\xi_1]_\sim, i \rangle$ and $\langle [\xi_2]_\sim, j \rangle$ are equivalent if $\xi_1 \sim \xi_2$ and $i = j$. Informally, $\mathcal{B}^2_{M, N}$ contains for each $\psi_i$ a copy of the state space required to verify $\ltlF (\psi_i[N]_\mu)$, indexed by $i$. Because such a formula is in $\mu$-\textsf{pLTL}, it is satisfied if and only if it is eventually rewritten by $\af[]$ into a formula propositionally equivalent to $\top$. When this happens, we transition to the next sub-component, which serves to verify $\ltlF (\psi_{i + 1}[N]_\mu)$. This continues until we reach the final component, at which point we loop back to the component with index $1$. It suffices to let $\langle \top , 1 \rangle$ be accepting, since the transition function ensures that this state is visited infinitely often precisely when $\langle \top , i \rangle$ is visited infinitely often for all $i \in [\vert M \vert]$. 

%Observe that the above construction is very similar to that used to ``degeneralize'' a generalized Büchi automaton. One difference is that we are not required to consider products of automata for each separate formula. This is not due to the particulars of our construction, but is a consequence of the fact that we require each formula in $M$ (rewritten by $\placeholder[N]_\mu$) to hold infinitely often. Moreover, our transition function must make use of the weakening conditions as tracked by $\mathcal{H}_\varphi$ when moving from one sub-component to the next. 

Lemma~\ref{lemma_B2} expresses the conditions under which $\mathcal{H}_\varphi \ltimes \mathcal{B}^2_{M, N}$ serves to verify the assumptions of Lemma~\ref{lemma_MN_=>}.
\begin{restatable}{lemma}{lemmabtwo}
\label{lemma_B2}
    Let $w$ be a word over $\mathcal{P}(\Var(\varphi))$, $t' \in \mathbb {N}$ and $i \in [\vert M \vert - 1]$. Let $\eta$ be the run of $\mathcal{H}_\varphi \ltimes \mathcal{B}^2_{M, N}$ on $w_{t'}$ starting from the state $\langle \beta, \gamma \rangle$, where $\beta = \cwc(S_0, w_{0t'})$ and $\gamma = \big\langle \iota(\ltlF \psi_i, \beta, N), i \big\rangle$. Then
\begin{itemize}
    \item 1) \, If $\FG{\varphi}{w}{0} \subseteq N$ and $w$ satisfies
    \begin{equation*}
    \forall \psi \in M \, . \, \forall t \geq 0 \, . \, w_t \models \ltlF(\psi\rw{\compseq{\varphi}{t}}[N\rw{\compseq{\varphi}{t}}]_\mu),
    \end{equation*}
    then $\eta$ is an accepting run.
    \item 2) \, If $N \subseteq \FG{\varphi}{w}{0}$ and $\eta$ is an accepting run, then
    \begin{equation*}
        \forall \psi \in M \, . \, \forall r \, . \, \exists t \geq r \, . \, w_t \models \ltlF (\psi\rw{\compseq{\varphi}{t}}[N\rw{\compseq{\varphi}{t}}]_\mu).
    \end{equation*}
\end{itemize}
\end{restatable}
\subsubsection{Verifying the Choice of \texorpdfstring{$N$}{N}}
\label{section_N_automaton}
The final automaton that we define before moving on to the full LDBA construction is the multiple-entry automaton $\mathcal{B}^3_{M, N} = (Q^3_{M, N}, \delta^3_{M, N}, \alpha^3_{M, N})$ over $\mathcal{P}(\Var(\varphi))$, with bed automaton $\mathcal{H}_\varphi$, where
\begin{align*}
    Q^3_{M, N} &\coloneqq \mathbb{B}^\lor\bigg(\varphi[M]_\nu \land \bigwedge_{\psi \in N} \ltlG(\psi[M]_\nu)\bigg)
    \\
    \delta^3_{M, N}(\zeta, \xi^\times, \sigma) &\coloneqq \af(\zeta)
    \\
    \alpha^3_{M, N} &\coloneqq Q^3_{M, N} \setminus \{\bot\}.
\end{align*}
The automaton $\mathcal{H}_\varphi \ltimes \mathcal{B}^3_{M, N}$ simultaneously verifies that all formulae in $N$ are satisfied persistently, under certain assumptions on $M$. Lemma~\ref{lemma_B3} expresses under what conditions $\mathcal{H}_\varphi \ltimes \mathcal{B}^3_{M, N}$ serves to verify the assumptions of Lemma~\ref{lemma_MN_<=}.

\begin{restatable}{lemma}{lemmabthree}
\label{lemma_B3}
Let $w$ be a word over $\mathcal{P}(\Var(\varphi))$ and $t \in \mathbb{N}$. Let $\eta$ be the run of $\mathcal{H}_\varphi \ltimes \mathcal{B}^3_{M, N}$ on $w_{t}$ starting from the state $\langle \beta, \gamma \rangle$, where $\beta = \cwc(S_0, w_{0t})$ and $\gamma = \iota(\ltlG \land_{\psi \in N} \,\psi, \beta, M)$. Then
\begin{itemize}
    \item 1) \, If $\GF{\varphi}{w}{0} \subseteq M$ and $w$ satisfies
    \begin{equation*}
    \forall \psi \in N \, . \, w_{t} \models \ltlG(\psi\rw{\compseq{\varphi}{t}}[M\rw{\compseq{\varphi}{t}}]_\nu),
    \end{equation*}
    then $\eta$ is an accepting run.
    \item 2) \, If $M \subseteq \GF{\varphi}{w}{0}$ and $\eta$ is an accepting run, then
    \begin{equation*}
        \forall \psi \in N \, . \, w_t \models \ltlG (\psi\rw{\compseq{\varphi}{t}}[M\rw{\compseq{\varphi}{t}}]_\nu).
    \end{equation*}
\end{itemize}
\end{restatable}
\subsection{Constructing the LDBA}
\label{section_constructing_ldba}
Let $\mathcal{H}_\varphi = (S, S_0, \delta_\mathcal{H})$. For each pair of sets of subformulas $M \subseteq \mu(\varphi)$ and $N \subseteq \nu(\varphi)$ we assume an ordering $\psi_1, \psi_2, \dots, \psi_{\vert M \vert}$ of the elements in $M$ and define three functions. These are transitions that describe for $i \in [3]$ the state reached in $B^i_{M, N}$ upon jumping from the non-deterministic part of the LDBA depending on the consumed prefix of the word being read and how it has transformed $\varphi$.
\begin{equation*}
\begin{aligned}
    Q^1_{0, M, N}(\zeta, \xi^\times, \sigma) &\coloneqq \iota(\af(\zeta), \xi^\times, M) \\
    Q^2_{0, M, N}(\zeta, \xi^\times, \sigma) &\coloneqq \big\langle \iota(\ltlF \psi_1, \xi^\times, N), 1 \big\rangle \\
    Q^3_{0, M, N}(\zeta, \xi^\times, \sigma) &\coloneqq \iota(\ltlG \land_{\psi \in N} \, \psi, \xi^\times, M).
\end{aligned}
\end{equation*}
The choice of index $1$ in the definition of $Q^2_{0, M, N}$ is arbitrary. We collect the components of the three automata defined in Section~\ref{section_ldba_components} and the above defined functions by taking their direct products:
\begin{equation*}
\begin{aligned}
    Q_{M, N} &\coloneqq Q^1_{M, N} \times Q^2_{M, N} \times Q^3_{M, N} \\
    Q_{0, M, N}(\zeta, \xi^\times, \sigma) &\coloneqq \big\langle
    Q^1_{0, M, N}(\zeta, \xi^\times, \sigma), Q^2_{0, M, N}(\zeta, \xi^\times, \sigma), Q^3_{0, M, N}(\zeta, \xi^\times, \sigma)
    \big\rangle
    \\
    \delta_{M, N}\big( \big\langle \zeta_1, \zeta_2, \zeta_3 \big\rangle, \xi^\times, \sigma \big) &\coloneqq \big\langle
    \delta^1_{M, N}(\zeta_1, \xi^\times, \sigma), \delta^2_{M, N}(\zeta_2, \xi^\times, \sigma), \delta^3_{M, N}(\zeta_3, \xi^\times, \sigma)
    \big\rangle \\
    \alpha_{M, N} &\coloneqq \alpha^1_{M, N} \times \alpha^2_{M, N} \times \alpha^3_{M, N}.
\end{aligned}
\end{equation*}
Following the intuition described in Section~\ref{section_ldba_overview}, we now define the runner automaton $\mathcal{A}_{LDBA}(\varphi) = (Q, Q_0, \delta, \alpha)$ with bed automaton $\mathcal{H}_\varphi$ with the following components:
\begin{equation*}
\begin{aligned}
    Q &\coloneqq \mathbb{B}^\lor(\varphi) \ \cup \ \bigcup \big\{ Q_{M, N} \mid M \subseteq \mu(\varphi), N \subseteq \nu(\varphi) \big\}\\
    Q_0 &\coloneqq \varphi \\
    \delta\big(\zeta, \xi^\times, \sigma\big) &\coloneqq \big\{ \af(\zeta) \big\} \ \cup \ \bigcup \big\{ Q_{0, M, N}(\zeta, \xi^\times, \sigma) \mid M \subseteq \mu(\varphi), N \subseteq \nu(\varphi) \big\} \\
    \delta\big(\big\langle \zeta_1, \zeta, \zeta_3 \big\rangle, \xi^\times, \sigma \big) &\coloneqq \delta_{M, N}\big(\big\langle \zeta_1, \zeta_2, \zeta_3 \big\rangle, \xi^\times, \sigma \big) \\
    \alpha &\coloneqq \bigcup \big\{ \alpha_{M, N} \mid  M \subseteq \mu(\varphi), N \subseteq \nu(\varphi) \big\}.
\end{aligned}
\end{equation*}
Observe that we define the transition function in a pattern matching style: the successor states of a state in $\mathbb{B}^\lor(\varphi)$ is given by $\delta(\zeta, \xi^\times, \sigma)$ and the successor state of a state in $Q_{M, N}$ is given by $\delta(\langle \zeta_1, \zeta_2, \zeta_3 \rangle, \xi^\times, \sigma)$.

The automaton begins in its non-deterministic part with state space $\mathbb{B}^\lor(\varphi)$, starting with the state labeled by $\varphi$. At each state and time $t$ it is possible to either keep updating $\varphi$ using the after-function, or to guess that the word is now stable with respect to $\varphi$ and jump to an arbitrary accepting component $Q_{M, N}$. In the latter case, we continue in the state $Q_{0, M, N}$, which is defined in terms of the initialization formulae; this ensures that the correct weakening conditions are added, as tracked by the WC automaton. We then proceed by verifying that $w_t \models \af[w_{0t}](\varphi)[M\rw{C}]_\nu$ for some $C$ whose weakening conditions are satisfied, using the automaton $\mathcal{B}^1_{M, N}$ described in Section~\ref{section_main_automaton}, as well as the assumptions on $M$ and $N$ using the automata described in Section~\ref{section_M_automaton} and Section~\ref{section_N_automaton}, respectively.

\begin{restatable}{theorem}{ldbacorrect}
    Let $\varphi$ be a formula. Then $\mathcal{L}(\mathcal{H}_\varphi \ltimes \mathcal{A}_{LDBA}(\varphi)) = \mathcal{L}(\varphi)$. Moreover, the number of states of $\mathcal{H}_\varphi \ltimes \mathcal{A}_{LDBA}(\varphi)$ is at most doubly exponential in the size of $\varphi$.
\end{restatable}

\section{From LDBA to DPA}
\label{section:dpa}
Esparza et al.~\cite{esparza2022dpa} present a translation from linear temporal logic (without past) to deterministic parity automata, via limit-deterministic Büchi automata. Provided that the LDBA automaton satisfies certain conditions, the translation is asymptotically optimal, resulting in an automaton with at most a doubly exponential number of states with respect to the size of the translated formula. We now establish that the construction of Section~\ref{section_pltl_to_ldba} satisfies these conditions. We begin with by providing the relevant definitions.
\begin{definition}[Initial-deterministic LDBA]
An LDBA $(Q, Q_0, \delta, \alpha)$ is initial-deterministic if the restriction $\delta \vert_{Q_{nd}}$ is deterministic. That is, for all $q \in Q_{nd}$ and letters $\sigma$, it is the case that $\vert \delta(q, \sigma) \cap Q_{nd} \vert = 1$.
\end{definition}
\begin{definition}[Bases]
A base of an LDBA $(Q, Q_0, \delta, \alpha)$ over an alphabet $\Sigma$ is a (finite) set $\mathcal{B} = \{ L_1, L_2, \dots, L_n \}$ of languages $L_i \subseteq \Sigma^\omega$, such that there, for every $q \in Q$, exists a set $I \subseteq [n]$ satisfying $\mathcal{L}(q) = \bigcup_{i \in I} L_i$.
\end{definition}
Intuitively, a base for an LDBA is a set of languages such that the set of words accepted from each state of the automaton is the union of some subset of languages from the base.
\begin{lemma}[\cite{esparza2022dpa}]
\label{lemma:ldba_property}
Let $\mathcal{A}$ be an initially-deterministic LDBA with a base $\mathcal{B}$. Then there exists an equivalent DPA with at most $2^{\mathcal{O}(b^2)}$ states, where $b = \vert \mathcal{B} \vert + 1$.
\end{lemma}

It is immediately clear that the automaton $\mathcal{H}_\varphi \ltimes \mathcal{A}_{LDBA}(\varphi)$ is initial-deterministic.
\begin{restatable}{lemma}{initdeterministic}
    Given a formula $\varphi$, the LDBA $\mathcal{A}_{LDBA}$ is initial deterministic.    
\end{restatable}

%Each state of the automaton is a product containing a linear number of formulae in $\mathbb{B}^\lor(\varphi')$, for some $\varphi'$. 
To motivate the existence of a sufficiently small base, we concentrate on formulae in $\mathbb{B}^{\lor}(\varphi')$, for some $\varphi'$ derived from $\varphi$ and of at most the same size, which form the transitions of the different components of the LDBA.
Each formula $\psi$ in $\mathbb{B}^{\lor}(\varphi')$ has the same language as some formula $\psi^d$ in $\mathbb{B}^{\lor}(\varphi')$ that is in disjunctive normal form. 
By construction, each disjunct of $\psi^d$ is a conjunctive clause belonging to $\mathbb{B}(\varphi'\rw{C})$, for some $C \subseteq \psf(\varphi')$. 
There are at most $2^{|\psf(\varphi)|}$ choices for $C\subseteq \psf(\varphi')$.
Each conjunctive clause, corresponds to a subset of $\sff(\varphi')$.
Hence, there are at most $2^{|\sff(\varphi)|+|\psf(\varphi)|}$ different disjuncts in $\psi^d$.
There are at most $2^{|\varphi|}$ conjunctive clauses in $\mathbb{B}(\varphi)$. 
% There are at most $2^{\vert \varphi \vert}$ conjunctive clauses in $\mathbb{B}(\varphi)$, up to initial equivalence. The set $\mathbb{B}^\lor(\varphi)$ is defined as disjunctions of conjunctive clauses in  Hence, the language of a formula in $\mathbb{B}^\lor(\varphi)$ is specified by some disjunct $\psi \in \mathbb{B}(\varphi\rw{C})$, where $C$ is the maximal set of past formulae whose weakening conditions are satisfied. 
Hence the language of $\psi$ is a union of at most $2^{\vert \varphi \vert} \cdot 2^{\vert \psf(\varphi)\vert }$ sub-languages.

Since the set of states of $\mathcal{H}_\varphi \ltimes \mathcal{A}_{LDBA}(\varphi)$ is defined in terms of a linear number of copies of the set $\mathbb{B}^\lor(\varphi)$, we obtain the following lemma:

\begin{restatable}{lemma}{expbase}
    Given a formula $\varphi$, the LDBA $\mathcal{A}_{LDBA}$ has a base of size $2^{\mathcal{O(\vert \varphi \vert)}}$.
\end{restatable}
From this, we immediately obtain the following corollary:
\begin{corollary}
    Let $\varphi$ be a formula. Then there exists a DPA $\mathcal{A}$ with a number of states at most doubly exponential in the size of $\varphi$, such that $\mathcal{L}(\mathcal{A}) = \mathcal{L}(\varphi)$, 
\end{corollary}

\section{Conclusion}
\label{section_discussion}
We have presented an asymptotically optimal translation from linear temporal logic with past to limit-deterministic Büchi automata, extending previous work by Esparza et al.~\cite{esparza_unified_translation}. The translation makes use of a framework for handling past operators by rewrites developed by Azzopardi et al.~\cite{azzopardi_et_al:LIPIcs.MFCS.2024.13}. In addition, we have shown that the resulting automata are amenable to a further asymptotically optimal translation to deterministic parity automata, due to Esparza et al.~\cite{esparza2022dpa}. The automata can be built incrementally on a by-need basis and do not require altering the states and transitions already created.
%\np{We note that our TNGBW construction, gives a fresh view on the classic translation of pLTL to NBW, which keeps a set of past formulae true at a given position on the side. Our construction internalizes these into the formulae themselves by weakening. It has the distinct advantage that past formulae that are not required anymore naturally disappear, while this can be optimized in the classical construction. We are not aware of this optimization being actually done in the past (pun intended). }
%\np{Guess not space for something like this.}

As future work, we would like to implement the above translations and explore whether the complexity involved in handling past operators is balanced by the added succinctness in practice. The translations by Esparza et al.~\cite{esparza_unified_translation} have been implemented in the tool Owl~\cite{owl}. It is of interest to explore how amenable those implementations are to our extensions. Owl is used by, among other things, the reactive synthesis tool Strix~\cite{strix} and the probabilistic model checker MoChiBA~\cite{mochiba}.
We would like to also study the usage of the natural encoding of the past in the conversion of \textsf{pLTL} to non-deterministic B\"uchi automata with transition-based acceptance, as used in state-of-the-art model-checking tools and in the extension of \textsf{pLTL} to include the `now' operator \cite{laroussinie2002temporal}.

\clearpage 

\bibliographystyle{eptcs}
\bibliography{main}

\clearpage
\end{document}